\documentclass[conference,letterpaper]{IEEEtran}

\usepackage[utf8]{inputenc} 
\usepackage[T1]{fontenc}
\usepackage{url}
\usepackage{ifthen}
\usepackage{cite}
\usepackage[cmex10]{amsmath} 
\usepackage{microtype}
\usepackage{graphicx}
\usepackage{booktabs} 

\usepackage{amsmath,amssymb,amsthm,textcomp,bm,xparse}
\usepackage{enumerate}
\usepackage{tikz}
\usepackage{cancel}
\usepackage{standalone}
\usepackage{float}

\usepackage{caption}
\usepackage{subcaption}

\newtheoremstyle{mytheor}
    {3ex}{3ex}{\em}{0ex}{\bfseries}{.}{2ex}
    {{\thmname{#1 }}{\thmnumber{#2}}{\thmnote{ (#3)}}}

\theoremstyle{mytheor}

\newtheorem{prop}{Proposition}

\begin{document}
\title{Mutual Information of Neural Network \\ Initialisations: Mean Field Approximations} 
\author{%
  \IEEEauthorblockN{Jared Tanner and Giuseppe Ughi}
  \IEEEauthorblockA{University of Oxford, Mathematical Institute\\
                    Andrew Wiles Building, Radcliffe Observatory Quarter\\
                    OX2 6GG, Oxford, UK\\
                    Email: \{tanner, ughi\}@maths.ox.ac.uk}
}

\maketitle
\IEEEpubidadjcol

\begin{abstract}
The ability to train randomly initialised deep neural networks is known to depend strongly on the variance of the weight matrices and biases as well as the choice of nonlinear activation. Here we complement the existing geometric analysis of this phenomenon  \cite{schoenholz2016deep} with an information theoretic alternative.  Lower bounds are derived for the mutual information between an input and hidden layer outputs.  Using a mean field analysis we are able to provide analytic lower bounds as functions of network weight and bias variances as well as the choice of nonlinear activation. These results show that initialisations known to be optimal from a training point of view are also superior from a mutual information perspective.

\end{abstract}

\section{Introduction}

Randomly initialised deep neural networks (DNNs) are random nonlinear functions which are drawn and subsequently trained to map a training set of inputs to known outputs. The work in \cite{schoenholz2016deep} showed that DNN initialisations that preserve information about the inputs are typically easier to train. This conclusion was based on geometric considerations of input signal dynamics in a random feed-forward DNN, measured with the distributions of intermediate hidden layers. This was possible as for a DNN denoted by 
\begin{align}\label{eq:FeedForward}
	\textbf{h}^{(\ell)} &=  \textbf{W}^{(\ell)}\phi\left(\textbf{h}^{(\ell-1)}\right) + \textbf{b}^{(\ell)}
\end{align}
where $\textbf{h}^{(1)} = \textbf{W}^{(1)}\textbf{X} + \textbf{b}^{(1)}$ with input $\textbf{X}\in\mathbb{R}^n$, and with $\textbf{W}^{(\ell)}\in\mathbb{R}^{n\times n}$, \cite{Poole2016} determined as a function of the DNN parameters $(\sigma_w,\sigma_b,\phi(\cdot))$ the dynamics of $q^{(\ell)}:=\|\textbf{h}^{(\ell)}\|^2_2$ to its large depth limit $q^*$ in the mean field infinite width limit for $n$ 
for Gaussian initialisation 
\begin{equation}\label{eq:random_wb}
\textbf{W}_{ij}^{(\ell)} \sim \mathcal{N}(0, \sigma_w^2/n)\quad\mbox{ and }\quad \textbf{b}_i^{(\ell)} \sim \mathcal{N}(0,\sigma_b^2).
\end{equation}
This allowed us to consider the geometric stability of the DNN when applied to two nearby points. Specifically, for a given nonlinear activation $\phi(\cdot)$ they derived the set of initialisation parameters $(\sigma_w,\sigma_b)$, denoted the edge of chaos (EoC), which separates the parameter space where nearby points converge to one another (ordered phase) from the domain where nearby points diverge (chaotic phase); see for example Figure \ref{fig:EOC1} for $\phi(\cdot)=\tanh(\cdot)$.  The EoC conditions were later shown by \cite{schoenholz2016deep} and \cite{Pennington2018} to similarly control the size of entries in the gradients used to train DNNs and are essential for training DNNs with many layers.

Here we conduct an alternative information theoretic investigation of random feed-forward DNNs in order to determine how the DNN parameters $(\sigma_w,\sigma_b,\phi(\cdot))$ impact the flow of information through the DNN. We derive the lower bound of the mutual information between the input $\textbf{X}$ and its associated hidden layer value $\textbf{h}^{(\ell)}$, denoted $I(\textbf{X},\textbf{h}^{(\ell)})$.  
Mutual Information (MI) is a measure of the dependence of two random variables.   Given two variables $(\textbf{X},\textbf{Y}) \in \mathcal{X}\times \mathcal{Y}$,  MI is defined as the Kullback–Leibler divergence between the joint distribution $P_{(\textbf{X},\textbf{Y})}$ and the marginal distributions $P_\textbf{X}$ and $P_\textbf{Y}$ 
\begin{align}
    I(\textbf{X}&,\textbf{Y}) = D_{KL}(P_{(\textbf{X},\textbf{Y})}||P_\textbf{X} \otimes P_\textbf{Y})\\ 
    &= \int_\mathcal{X}\int_\mathcal{Y} p_{\textbf{X},\textbf{Y}}(\textbf{x},\textbf{y})\log \left( \frac{p_{\textbf{X},\textbf{Y}}(\textbf{x},\textbf{y})}{p_{\textbf{X}}(\textbf{x})p_{\textbf{Y}}(\textbf{y})}\right)\ d\textbf{x}\ d\textbf{y}.
\end{align}

Our main contribution are lower bounds on $I(\textbf{X},\textbf{h}^{(\ell)})$ which, similarly to \cite{schoenholz2016deep}, are functions of the DNN parameters $(\sigma_w,\sigma_b,\phi(\cdot))$.  We include the mean field analysis of \cite{Pennington2018} in order to obtain an analytic approximation of the lower bounds on $I(\textbf{X},\textbf{h}^{(\ell)})$ and observe that  $I(\textbf{X},\textbf{h}^{(\ell)})$ is maximised on the EoC, thus suggesting that initialisations which are optimal for geometric training analysis are also preferable from an MI perspective.

\begin{figure}[t!]
\centering
\includegraphics[width=.7\linewidth, trim={0.5cm 0.22cm 1.35cm 1cm},clip]{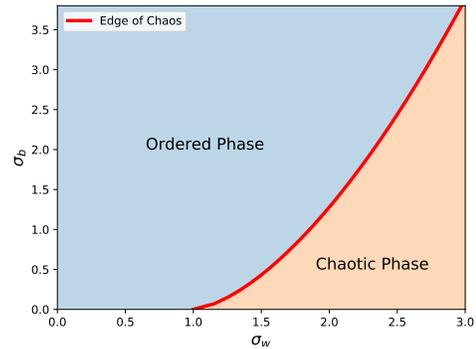}
  \caption{Relation between $\sigma_b$ and $\sigma_w$ at the Edge of Chaos for a feed-forward DNN with $\phi(\cdot) = tanh(\cdot)$. For combinations within the ordered phase, the similar inputs converge to the same output and the gradients vanish with depth, while within the chaotic phase similar inputs diverge and the gradients explode.}
  \label{fig:EOC1}
\end{figure}


In Section \ref{sec:BoundsOnMI}, we prove the existence of a lower bound on the MI problem by introducing the Gaussian model of the DNN. In Section \ref{sec:approximations_lb}, we approximate this lower bound via the mean field approach. Finally in Section \ref{sec:Gabrie}, we compare our lower bound approximation to a state-of-the-art MI approximator and see that the initialisation at edge of chaos similarly increases the MI at deep layers.

 \section{Gaussian Lower Bound on the MI}\label{sec:BoundsOnMI}

A MI analysis of DNNs was previously conducted in the unsupervised setting \cite{Linsker} and furthermore in the analysis of deep learning architectures \cite{Shwartz-Ziv,Goldfeld2018,saxe2019information,gabrie2019entropy,Foggo2020}.  For example, the work in \cite{Shwartz-Ziv} proposed the use of MI to describe the state of the training of DNNs by plotting the MI between an input $\textbf{X}$ and the hidden layer $\textbf{h}^{(\ell)}$, $I(\textbf{X},\textbf{h}^{(\ell)})$ against the MI between the hidden layer $\textbf{h}^{(\ell)}$ and the output $\textbf{Y}$, $I(\textbf{h}^{(\ell)},\textbf{Y})$.  The focus of \cite{Shwartz-Ziv} is on the dynamics of the training process, as opposed to the dependence on DNN parameters $(\sigma_w,\sigma_b,\phi(\cdot))$ at initialisation.  In \cite{Abrol2020}, it was observed experimentally that the MI planes advocated by \cite{Shwartz-Ziv} depend strongly on the parameters $(\sigma_w,\sigma_b,\phi(\cdot))$.   The mathematical analysis here compliments the observations in \cite{Abrol2020}.

We compute the MI $I(\textbf{X},\textbf{h}^{(\ell)})$ between the input $\textbf{X}\sim\mathcal{N}(\textbf{0},\sigma_x^2\textbf{I})$, and the hidden layer $\textbf{h}^{(\ell)}$ by conditioning on a realisation of the random weights ${\mathcal{W}^{:\ell} = \{\textbf{W}^{(i)},\textbf{b}^{(i)}, \textbf{n}^{(i)} \}_{i=1}^\ell}$
\begin{align}\label{eq:expected_MI}
     I(\textbf{X},\textbf{h}^{(\ell)}) &= \mathbb{E}_{\mathcal{W}^{:l}}\left[  I(\textbf{X},\textbf{h}^{(\ell)}| \mathcal{W}^{:\ell})\right]
\end{align}
where $\textbf{n}^{(\ell)}$ is a Gaussian noise variable that we add before the activation function as shown in Figure \ref{fig:noisy_network}. This noise term is necessary because once the the DNN map is sampled, $\textbf{X}\mapsto \textbf{h}^{(\ell)}$ is completely deterministic and consequently, the MI $I(\textbf{X},\textbf{h}^{(\ell)}| \mathcal{W}^{:\ell})$ between the hidden layers and the input is infinite. Initially, this was solved by considering binning the values in the hidden variables \cite{Shwartz-Ziv}, but the results were found to be too dependent on the choices of the bins \cite{saxe2019information}. The research in \cite{Goldfeld2018} and \cite{Foggo2020} showed that it is more appropriate to consider a random noise $\textbf{n}^{(\ell)}\sim \mathcal{N}(0, \sigma_n^2)$ added at each layer.

\begin{figure}[h]
    \centering
    \vspace{0.5cm}
    \includegraphics[width=0.75\linewidth,trim={0.1cm 0.35cm 0.18cm 0.03cm},clip]{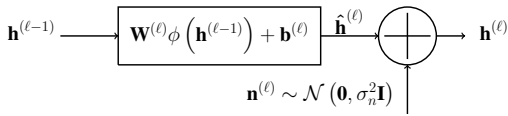}
    \caption{Signal Propagation in the considered perturbed DNN.}
    \label{fig:noisy_network}
\end{figure}

For two random variables with a generic distribution, the MI is not known explicitly and its approximation is a challenging task that has been attempted primarily with non-parametric models \cite{kolchinsky2018caveats,gabrie2019entropy,noshad2019scalable}. However, if two random variables $\textbf{x}$ and $\textbf{y}$ are Gaussian with marginal covariances, $\bm{\Lambda}_x$ and $\bm{\Lambda}_{y}$, and joint covariance, $\bm{\Lambda}_{xy}$, the MI is 
\begin{align}
    GMI(\textbf{X},\textbf{Y}) &= \frac{1}{2}\log \left(\frac{\left| \bm{\Lambda}_x \right| \left| \bm{\Lambda}_{y}\right|}{\left| \bm{\Lambda}_{xy} \right|}  \right) \label{eq:mut_inf_det}
\end{align}
While the distribution of $\textbf{h}^{(\ell)}$ is  known to converge to a Gaussian distribution with large depth $\ell$, in order to lower bound the MI throughout the layers we note that the MI for a general distribution is lower bounded by that of a Gaussian distribution.

\begin{prop}\label{prop:f_to_g}
Let g(x,y) be an n-dimensional Gaussian distribution , $\mathcal{N}(\mu, \Lambda)$, with mean $\mu$ and covariance matrix $\Lambda$. If f(x,y) is an arbitrary distribution with the same mean and covariance matrix, and with f(x) Gaussian, then
\begin{equation}
    I_{f_{(xy)}}[X,Y] \geq I_{g_{(xy)}}[X,Y]
\end{equation}
\end{prop}
\begin{proof}
\fontsize{10}{10}\selectfont
\begin{align}\nonumber
    &I_{f}[X,Y] - I_{g}[X,Y] =  \int\limits_\mathcal{X,Y} f_{\textbf{X},\textbf{Y}}(\textbf{x},\textbf{y})\log \left( \frac{f_{\textbf{X},\textbf{Y}}(\textbf{x},\textbf{y})}{f_{\textbf{X}}(\textbf{x})f_{\textbf{Y}}(\textbf{y})}\right)d\textbf{x}d\textbf{y} \\
        & \ -  \int_\mathcal{X,Y} g_{\textbf{X},\textbf{Y}}(\textbf{x},\textbf{y})\log \left( \frac{g_{\textbf{X},\textbf{Y}}(\textbf{x},\textbf{y})}{g_{\textbf{X}}(\textbf{x})g_{\textbf{Y}}(\textbf{y})}\right)d\textbf{x}d\textbf{y} \nonumber\\
    &= - \int_\mathcal{X} f_\textbf{X}(\textbf{x})\log \left( f_{\textbf{X}}(\textbf{x})\right)d\textbf{x} -
    \int_\mathcal{Y} f_\textbf{Y}(\textbf{y})\log \left( f_{\textbf{Y}}(\textbf{y})\right)d\textbf{y} \nonumber\\
        & \ + \int\limits_\mathcal{X,Y} f_{\textbf{X}\textbf{Y}}(\textbf{x,y})\log \left( f_{\textbf{X}\textbf{Y}}(\textbf{x,y})\right)d\textbf{x}d\textbf{y} +\int\limits_\mathcal{X} g_\textbf{X}(\textbf{x})\log \left( g_{\textbf{X}}(\textbf{x})\right)d\textbf{x}\nonumber\\
        & \ + \int\limits_\mathcal{Y} g_\textbf{Y}(\textbf{y})\log \left( g_{\textbf{Y}}(\textbf{y})\right)d\textbf{y} - 
        \int\limits_\mathcal{X,Y} g_{\textbf{X}\textbf{Y}}(\textbf{x,y})\log \left( g_{\textbf{X}\textbf{Y}}(\textbf{x,y})\right)d\textbf{x}d\textbf{y}\nonumber\\
    &= - \int\limits_\mathcal{Y} f_\textbf{Y}(\textbf{y})\log \left( f_{\textbf{Y}}(\textbf{y})\right)d\textbf{y} +
    \int\limits_\mathcal{X,Y} f_{\textbf{X}\textbf{Y}}(\textbf{x,y})\log \left( f_{\textbf{X}\textbf{Y}}(\textbf{x,y})\right)d\textbf{x}d\textbf{y}\nonumber\\
        & \  + \int\limits_\mathcal{Y} g_\textbf{Y}(\textbf{y})\log \left( g_{\textbf{Y}}(\textbf{y})\right)d\textbf{y}-
        \int\limits_\mathcal{X,Y} g_{\textbf{X}\textbf{Y}}(\textbf{x,y})\log \left( g_{\textbf{X}\textbf{Y}}(\textbf{x,y})\right)d\textbf{x}d\textbf{y}\nonumber\\
    &= \int_\mathcal{Y} f_\textbf{Y}(\textbf{y})\log \left(\frac{g_{\textbf{Y}}(\textbf{y})}{f_{\textbf{Y}}(\textbf{y})} \right)d\textbf{y}\nonumber\\
    &\quad -  
        \int_\mathcal{X,Y} f_{\textbf{X}\textbf{Y}}(\textbf{x,y})\log \left( \frac{g_{\textbf{X}\textbf{Y}}(\textbf{x,y})}{f_{\textbf{X}\textbf{Y}}(\textbf{x,y})}\right)d\textbf{x}d\textbf{y}\nonumber\\
    &= \int_\mathcal{X,Y} f_{\textbf{X}\textbf{Y}}(\textbf{x},\textbf{y})\log \left(\frac{g_{\textbf{Y}}(\textbf{y})}{f_{\textbf{Y}}(\textbf{y})} \right)d\textbf{x}d\textbf{y} \nonumber\\
        & \quad -\int_\mathcal{X,Y} f_{\textbf{X}\textbf{Y}}(\textbf{x,y})\log \left( \frac{g_{\textbf{X}\textbf{Y}}(\textbf{x,y})}{f_{\textbf{X}\textbf{Y}}(\textbf{x,y})}\right)d\textbf{x}d\textbf{y} \nonumber\\
    & = \int_\mathcal{X,Y} f_{\textbf{X}\textbf{Y}}(\textbf{x},\textbf{y})\log \left(\frac{g_{\textbf{Y}}(\textbf{y})}{f_{\textbf{Y}}(\textbf{y})} \frac{f_{\textbf{X}\textbf{Y}}(\textbf{x,y})}{g_{\textbf{X}\textbf{Y}}(\textbf{x,y})}\right)\ d\textbf{x}d\textbf{y}\nonumber\\
    & = \int_\mathcal{Y} f_\textbf{Y}(\textbf{y}) \int_\mathcal{X} f_{\textbf{X}|\textbf{Y}}(\textbf{x}|\textbf{Y}=\textbf{y})\log \left( \frac{f_{\textbf{X}|\textbf{Y}}(\textbf{x}|\textbf{Y}=\textbf{y})}{g_{\textbf{X}|\textbf{Y}}(\textbf{x}|\textbf{Y}=\textbf{y})}\right)\ d\textbf{x}d\textbf{y} \nonumber\\
    & = \int_\mathcal{Y} f_\textbf{Y}(\textbf{y}) D_{KL}(f_{\textbf{X}|\textbf{Y=y}} || g_{\textbf{X}|\textbf{Y=y}})d\textbf{y}  \geq 0 \nonumber
\end{align}
\textit{where the third equality is due to $f_\textbf{x}$ being Gaussian, hence $f_\textbf{x} = g_\textbf{x}$; the fourth is due to \cite{hot} where it is shown that for the distributions $f$ and $g$ considered here }
\begin{equation}
    \int_\mathcal{X}f(\textbf{x})\log(g(\textbf{x})d\textbf{x} = \int_\mathcal{X}g(\textbf{x})\log(g(\textbf{x}))d\textbf{x} \nonumber
\end{equation}
\end{proof}

The analysis herein relies on the first two moments of $\begin{bmatrix}
    \textbf{X}^\top,\;\;
    \textbf{h}^{(\ell)\top}
    \end{bmatrix}^T$
which for DNNs drawn according to \eqref{eq:random_wb} with $\textbf{X}\sim\mathcal{N}(\textbf{0},\sigma_x^2\textbf{I})$ are:
\begin{align}
    \mathbb{E}_{\textbf{X}|\mathcal{W}^{:\ell}}
    \begin{bmatrix}
    \textbf{X}\\
    \textbf{h}^{(\ell)}
    \end{bmatrix}
    &= 
    \begin{bmatrix}
        \textbf{0}\\
        \bm{\mu}
    \end{bmatrix} \label{eq:norm_mean}
    \\
    \bm{\Lambda}_{xh^{(\ell)}}:=\mathbb{V}ar_{\textbf{X}|\mathcal{W}^{:\ell}}
    \begin{bmatrix}
    \textbf{X}\\
    \textbf{h}^{(\ell)}
    \end{bmatrix}
    &=
    \begin{bmatrix}
    \sigma_x^2\textbf{I} & \bm{\Sigma}_{xh^{(\ell)}}\\
    \bm{\Sigma}_{xh^{(\ell)}}^\top & \bm{\Lambda}_{h^{(\ell)}}
    \end{bmatrix}. \label{eq:normal_distribution}
\end{align}

In order to compute the GMI in \eqref{eq:mut_inf_det}, we reformulate the determinant of the block covariance matrix \eqref{eq:normal_distribution} using the decomposition  \cite{strang1993introduction}:
\begin{align}\label{eq:block_det}
    \begin{vmatrix}
        \bm{\Lambda}_x & \bm{\Lambda}_{xy}\\
        \bm{\Lambda}_{xy}^\top & \bm{\Lambda_y}
    \end{vmatrix}
    =
    \left|\bm{\Lambda}_x\right|
    \left|\bm{\Lambda}_y-\bm{\Lambda}_{xy}^\top\bm{\Lambda}_x^{-1}\bm{\Lambda}_{xy}
    \right|.
\end{align}

Incorporating \eqref{eq:block_det} for 
the determinant of $\bm{\Lambda}_{xh^{(l)}}$ in \eqref{eq:normal_distribution} we define the following lower bound for the MI for a DNN conditional on a set of weights $\mathcal{W}^{:\ell}$ with Gaussian input
\begin{align}
    I(\textbf{X},&\textbf{h}^{(\ell)}| \mathcal{W}^{:\ell}) &\geq \frac{1}{2}\log \left(\frac{ \left| \bm{\Lambda}_{h^{(\ell)}}\right|}{\left| \bm{\Lambda}_{h^{(\ell)}} -\frac{1}{\sigma_x^2} \bm{\Sigma}_{xh^{(\ell)}}^\top\bm{\Sigma}_{xh^{(\ell)}} \right|}  \right).
\end{align}

Thus, from \eqref{eq:expected_MI} the MI of a DNN is bounded as follows
\begin{align}
     I(\textbf{X},\textbf{h}^{(\ell)})
        & \geq \frac{1}{2}\underbrace{\mathbb{E}_{\mathcal{W}^{:\ell}}\left[ \log \left( \left| \bm{\Lambda}_{h^{(\ell)}}\right|  \right)\right]}_{(E_1)} \label{eq:e1} \\
        & \quad \underbrace{- \mathbb{E}_{\mathcal{W}^{:\ell}}\left[ \frac{1}{2}\log \left(\left| \bm{\Lambda}_{h^{(\ell)}} -\frac{1}{\sigma_x^2} \bm{\Sigma}_{xh^{(\ell)}}^\top\bm{\Sigma}_{xh^{(\ell)}} \right|  \right)\right]}_{(E_2)} . \label{eq:e2}
\end{align}

\subsection{Numerical Evaluation of the Gaussian Lower Bound}

The quantities ($E_1$) and ($E_2$) in  \eqref{eq:e1} and \eqref{eq:e2} respectively can be computed by sampling different realisations of the weights $\mathcal{W}^{:\ell}$ and by then averaging over the log-determinant of the matrices $\bm{\Lambda}_{h^{(\ell)}}$ and $\bm{\Lambda}_{h^{(\ell)}} -\frac{1}{\sigma_x^2} \bm{\Sigma}_{xh^{(\ell)}}^\top\bm{\Sigma}_{xh^{(\ell)}}$.  Here we compute these estimates by sampling $10^5$ different inputs $\textbf{X}$ and noises $\textbf{n}$ with $\sigma_x=1$ and $\sigma_n=0.1$. Figure \ref{fig:num_gaussian_lb} shows how the Gaussian lower bounds changes as a function of $\sigma_w$ for DNNs with square weight matrices of size $n\times n$ with $n=90$, with the $tanh$ activation function, and with $\sigma_b$ chosen such that ($\sigma_w$, $\sigma_b$) lies on the EoC.

\begin{figure}
    \centering
    \includegraphics[width=0.8\linewidth,trim={0.3cm 0.1cm 1.4cm 1cm},clip]{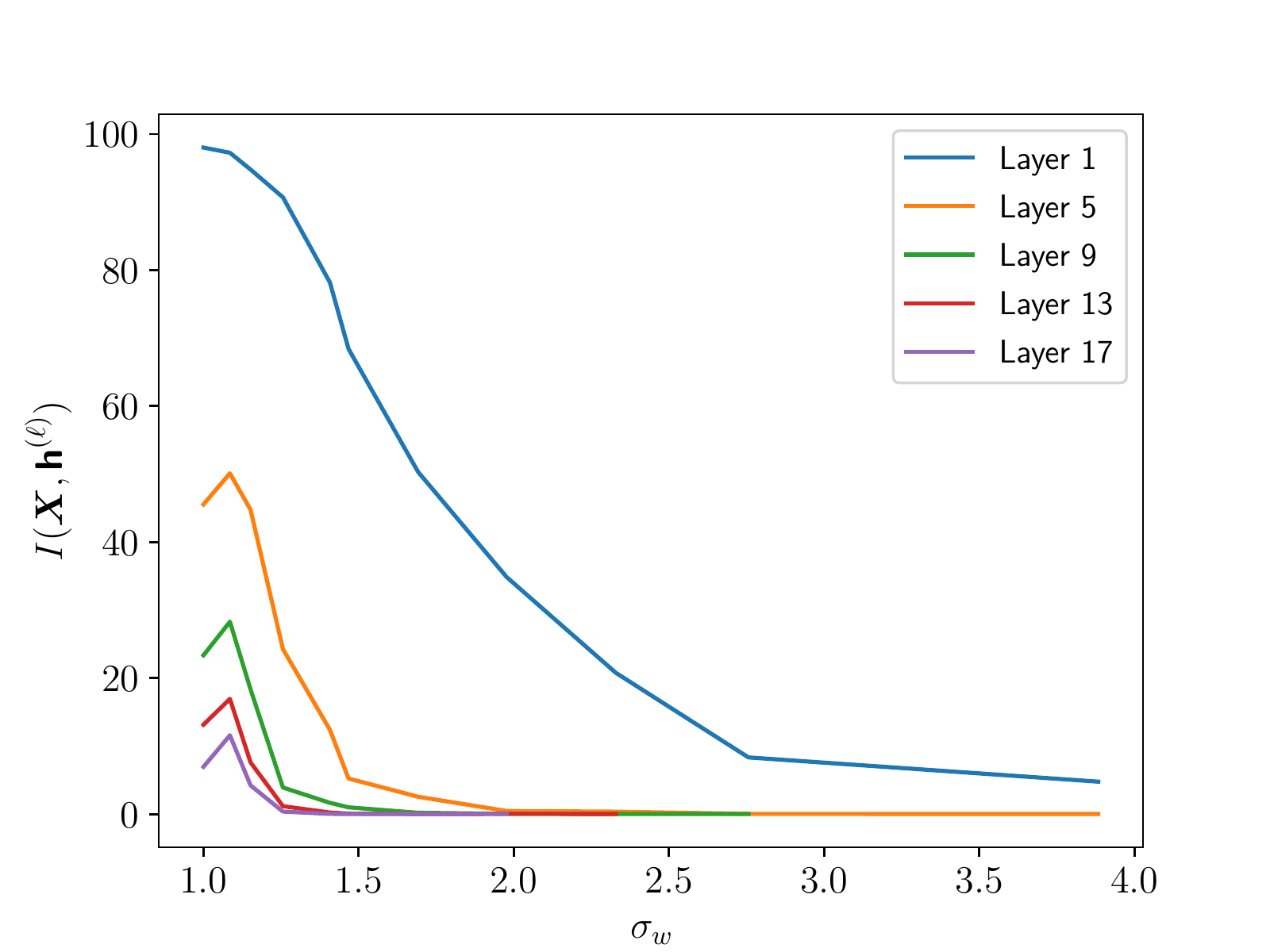}
    \caption{Numerical approximation by sampling of the Gaussian lower bound in \eqref{eq:e1}-\eqref{eq:e2} for the conditional MI of a feed-forward DNN when $\sigma_w$ varies along the EoC with $n=90$ and $\phi(\cdot)=tanh(\cdot)$.}
    \label{fig:num_gaussian_lb}
\end{figure}

These numerical experiments show how the MI between the input and the hidden layers decreases at each layer and is maximised for a value of $(\sigma_w,\sigma_b)$ close to $(1,0)$, as advocated in \cite{Poole2016}.


\section{An Analytic Gaussian Lower Bound}\label{sec:approximations_lb}

The numerical computation of ($E_1$) and ($E_2$) based on sampling the covariance matrix is computationally expensive for large layer width $n$ and depth $L$. Moreover, it less directly shows how this observed MI lower bound compares with the mean field analysis in \cite{Poole2016} and the corresponding EoC analysis.  For easier computation and to better link these analyses, we approximate the  matrices $\bm{\Lambda}_{h^{(l)}}$ and $\bm{\Lambda}_{h^{(l)}} - \frac{1}{\sigma_x^2}\bm{\Sigma}_{xh^{(l)}}^\top\bm{\Sigma}_{xh^{(l)}}$ based on the mean field assumption of \cite{Poole2016}. 

\subsection{Mean Field Approximation}

By the mean field analysis in \cite{Poole2016}, the hidden layers are normally distributed in the large limit, i.e. $n\gg1$, with the mean and variance given in \eqref{eq:norm_mean} and \eqref{eq:normal_distribution} respectively. In order to compute \eqref{eq:e1} and \eqref{eq:e2}, we model the factors $\bm{\Lambda}_{h^{(l)}}$ and $\bm{\Lambda}_{h^{(l)}} - \frac{1}{\sigma_x^2}\bm{\Sigma}_{xh^{(l)}}^\top\bm{\Sigma}_{xh^{(l)}}$ respectively by their expectation over the weights $\mathcal{W}^{:\ell}$.


\subsubsection{Expectation of $\bm{\Lambda}_{h^{(\ell)}}$}

\begin{figure}[t!]
 \centering
    \begin{subfigure}{.45\textwidth}
      \centering
      \includegraphics[width=.99\linewidth]{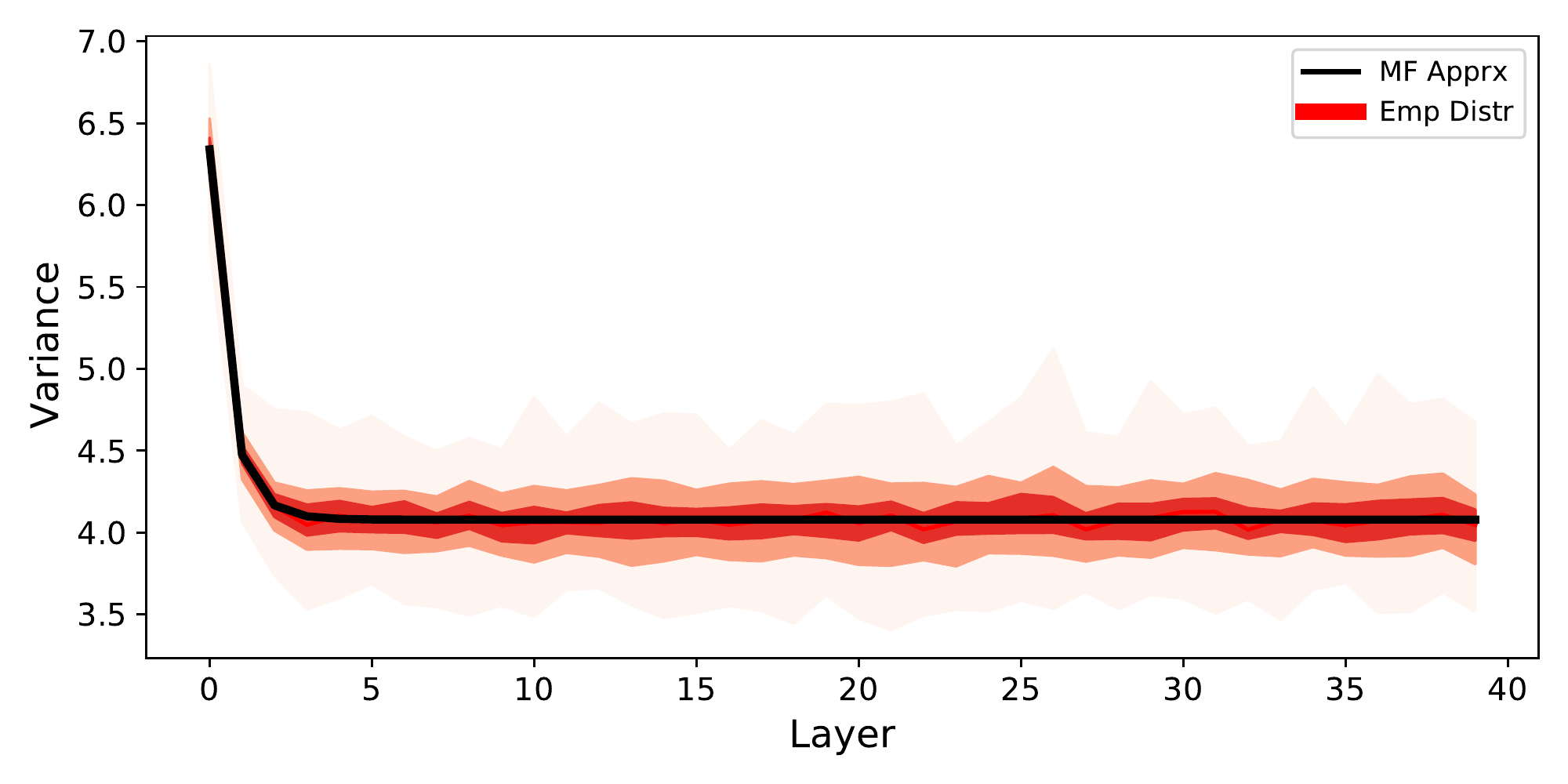}
      \caption{}
      \end{subfigure}%
     \\
    \begin{subfigure}{.45\textwidth}
    \centering
      \includegraphics[width=0.99\linewidth]{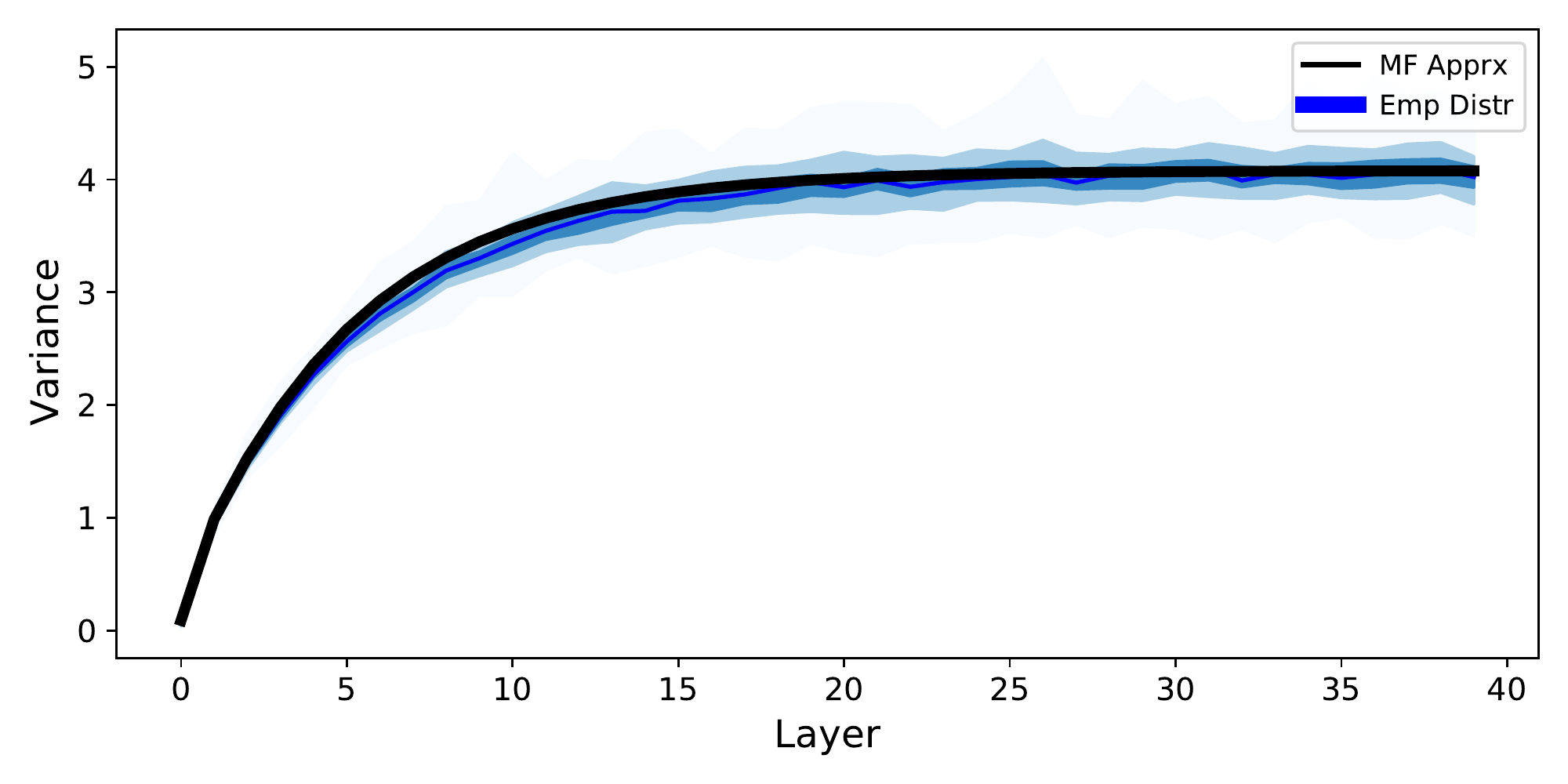}
      \caption{}
    \end{subfigure}
    \centering
    \caption{In (a) and (b) we compare mean field approximation (MF Apprx) of the mean value on the diagonal of the matrices $\bm{\Lambda}_{h^l}$ and ${\bm{\Lambda}_{h^{(\ell)}} - \frac{1}{\sigma_x^2}\bm{\Sigma}_{xh^{(\ell)}}^\top\bm{\Sigma}_{xh^{(\ell)}}}$ respectively with the empirical distribution (Emp. Distr) in function of the layers. The empirical distribution is obtained by considering 100 different set of weights $\mathcal{W}^{:\ell}$ for the DNN when $n=50$ and with $10^5$ different random inputs.}
    \label{fig:covariance_convergence}
\end{figure}

The expectation of  $\bm{\Lambda}_{h^l}$ in \eqref{eq:e1} was studied in \cite{Poole2016, Pennington2018} and was shown that for large DNN width $n\gg 1$ 
\begin{equation}\label{eq:expectation_num}
    \mathbb{E}_{\mathcal{W}^{:\ell}}\left[\bm{\Lambda}_{h^{(\ell)}}\right] = q^{(\ell)} \textbf{I},
\end{equation}
where $q^{(\ell)}$ corresponds to the variance of the hidden layer signal and is defined recursively via
\begin{align}
    q^{(\ell)} &= \sigma_w^2\int \phi(\sqrt{q^{(\ell-1)}}z)\mathcal{D}z + \sigma_b^2 + \sigma_n^2\\
    q^{(1)} &= \sigma_w^2\sigma_x^2 + \sigma_b^2 + \sigma_n^2.
\end{align}
In Figure \ref{fig:covariance_convergence}a we demonstrate the validity of this approximation by plotting $q^{(\ell)}$ as the solid black line along with the empirical distributions obtained by generating 100 realisations of \eqref{eq:FeedForward} for $(\sigma_w,\sigma_b)=(2.5,0.3)$, $\phi(\cdot)=tanh(\cdot)$ and $n=50$ and showing the distribution over $10^4$ randomly drawn inputs $\textbf{X}$.  Note the good agreement of $q^{(l)}$ and the empirical values, as well as the limiting value for large depth, denoted $q^*$ in \cite{schoenholz2016deep}.  The striking agreement is notable given the relatively small DNN width $n=50$ and the mean field approximation following from the $n\rightarrow\infty$ limit.

\subsubsection{Expectation of $\bm{\Lambda}_{h^{(\ell)}} - \frac{1}{\sigma_x^2}\bm{\Sigma}_{xh^{(\ell)}}^\top\bm{\Sigma}_{xh^{(\ell)}}$} To compute the expectation of the argument in the logarithm in ($E_2$) from \eqref{eq:e2}, we can make use of \eqref{eq:expectation_num} and then compute $\mathbb{E}_{\mathcal{W}^{:\ell}}\left[\bm{\Sigma}_{xh^{(\ell)}}^\top\bm{\Sigma}_{xh^{(\ell)}}\right]$. Since
\begin{equation}
    \mathbb{E}_{\mathcal{W}^{:\ell}}\left[\bm{\Sigma}_{xh^{(\ell)}}^\top\bm{\Sigma}_{xh^{(\ell)}}\right]_{ij} = \mathbb{E}_{\mathcal{W}^{:\ell}}\left[\sum_k \mathbb{E}\left[\textbf{x}_k \textbf{h}_i^{(\ell)}\right]\mathbb{E}\left[\textbf{x}_k \textbf{h}_j^{(\ell)}\right]\right] \nonumber
\end{equation}
and as the weights in $\textbf{W}^{(\ell)}$ are independent, the expected covariance matrix is diagonal with values 
\begin{equation}
    \mathbb{E}_{\mathcal{W}^{:\ell}}\left[\bm{\Sigma}_{xh^{(\ell)}}^\top\bm{\Sigma}_{xh^{(\ell)}}\right]_{ii} = \mathbb{E}_{\mathcal{W}^{:\ell}}\left[\sum_k \mathbb{E}\left[\textbf{x}_k \textbf{h}_i^{(\ell)}\right]^2\right] .
\end{equation}
With the mean field analysis, we consider any elements $i$ and $j$ of respectively the input and the hidden layer to be jointly normally distributed according to \begin{equation}
\begin{bmatrix}
\textbf{x}_i\\
\textbf{h}_j^{(\ell)}
\end{bmatrix}
\sim 
\mathcal{N}\left(\begin{bmatrix}
0\\
0
\end{bmatrix},\begin{bmatrix}
\sigma_x^2 & \rho^{(\ell)}\sigma_x \sqrt{q^{(\ell)}}\\
 \rho^{(\ell)}\sigma_x \sqrt{q^{(\ell)}} & q^{(\ell)}
\end{bmatrix}
\right)
\end{equation}
with $\rho^{(\ell)}$ being the correlation coefficient.  The correlation is given at each layer by solving 
\begin{align}
	 \rho^{(\ell)}&\sigma_x \sqrt{q^{(\ell)}} = \mathbb{E}_{\mathcal{W}^{:\ell}}\left[ \textbf{x}_i\textbf{h}_j^{(\ell)}\right] = \int \int   u_1 u_2 \mathcal{D}u_1 \mathcal{D}u_2\\
    &= \underbrace{\int \int \sigma_x z_1 \frac{\sigma_w}{\sqrt{n}}\phi\left( \sqrt{q^{(\ell)}}\left(\rho^{(\ell-1)} z_1 + \sqrt{1-\rho^{(\ell-1)2}} z_2 \right) \right)}_{(E_3)}
\end{align}
The expected squared covariance is then given by 
\begin{align}
	\mathbb{E}_{\mathcal{W}^{:\ell}}\left[\sum_k \mathbb{E}\left[\textbf{x}_k \textbf{h}_i^{(\ell)}\right]^2\right] = n(E_3)^2\textbf{I}
\end{align}
and consequently
\begin{equation}\label{eq:qc}
 \mathbb{E}\left[\bm{\Lambda}_{h^{(\ell)}} - \frac{1}{\sigma_x^2}\bm{\Sigma}_{xh^{(\ell)}}^\top\bm{\Sigma}_{xh^{(\ell)}}\right]  = (q^{(\ell)} - n*(E_3)^2/\sigma_x^2) \textbf{I} = q_c^{(\ell)} \textbf{I}.
\end{equation}
The validity of this approximation is shown in Figure \ref{fig:covariance_convergence}b were as before we note the excellent agreement of the mean field limit \eqref{eq:qc} corresponding to width $n\rightarrow\infty$ and the observed distribution for \eqref{eq:FeedForward} with width $n=50$.

\subsubsection{Lower Bound Approximation}

Since the matrices $\bm{\Lambda}_{h^{(l)}}$ and $\bm{\Lambda}_{h^{(l)}} - \frac{1}{\sigma_x^2}\bm{\Sigma}_{xh^{(l)}}^\top\bm{\Sigma}_{xh^{(l)}}$ are approximated as multiples of the identity matrix, their log-determinant can be easily computed; in particular

\begin{prop}\label{prop:mf_lowerbound}
Under the mean field approximation the MI has the following lower bound 
\begin{align}\label{eq:lower_bound}
    I&(\textbf{X},\textbf{h}^{(\ell)}) \geq \frac{n}{2}\log\left(\frac{q^{(\ell)}}{q_c^{(\ell)}}\right)
\end{align}
\end{prop}

\begin{figure}[t!]
    \centering
    \includegraphics[width=0.4\textwidth,trim={0.3cm 0.1cm 1.4cm 1cm},clip]{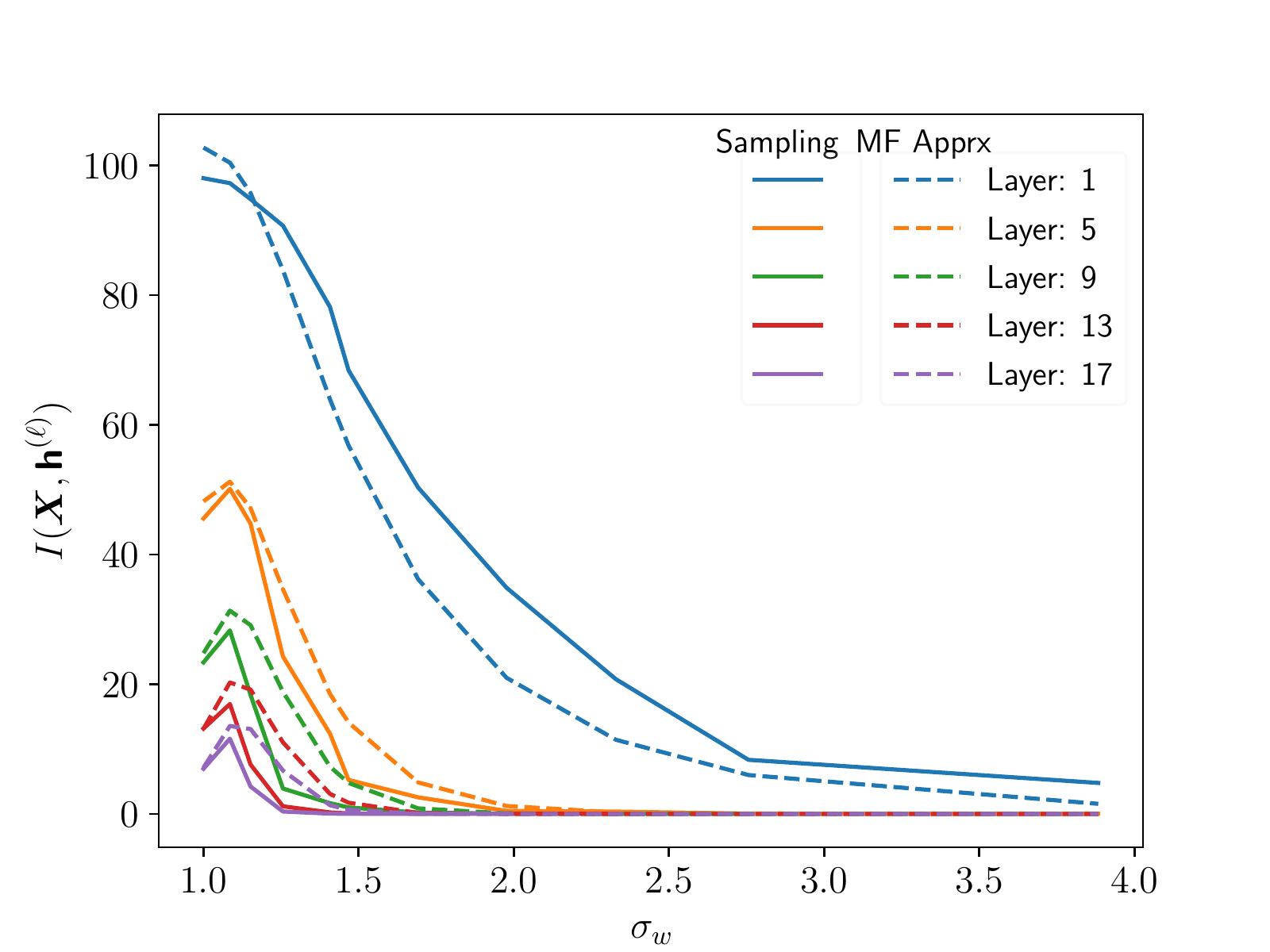}
\caption{Comparison of the analytical lower bound \eqref{eq:lower_bound} from the mean field approximation to the sample MI obtained for a DNN with $n=90$ and $tanh(\cdot)$ activation function when the variance $\sigma_w$ is changed. Here $\sigma_b$ is update so that $(\sigma_w, \sigma_b)$ satisfy the EoC condition.}
\label{fig:approx}
\end{figure}

\begin{figure}[t!]
\centering
  \begin{subfigure}{.4\textwidth}
    \includegraphics[width=\textwidth,trim={0.3cm 0.1cm 1.4cm 1cm},clip]{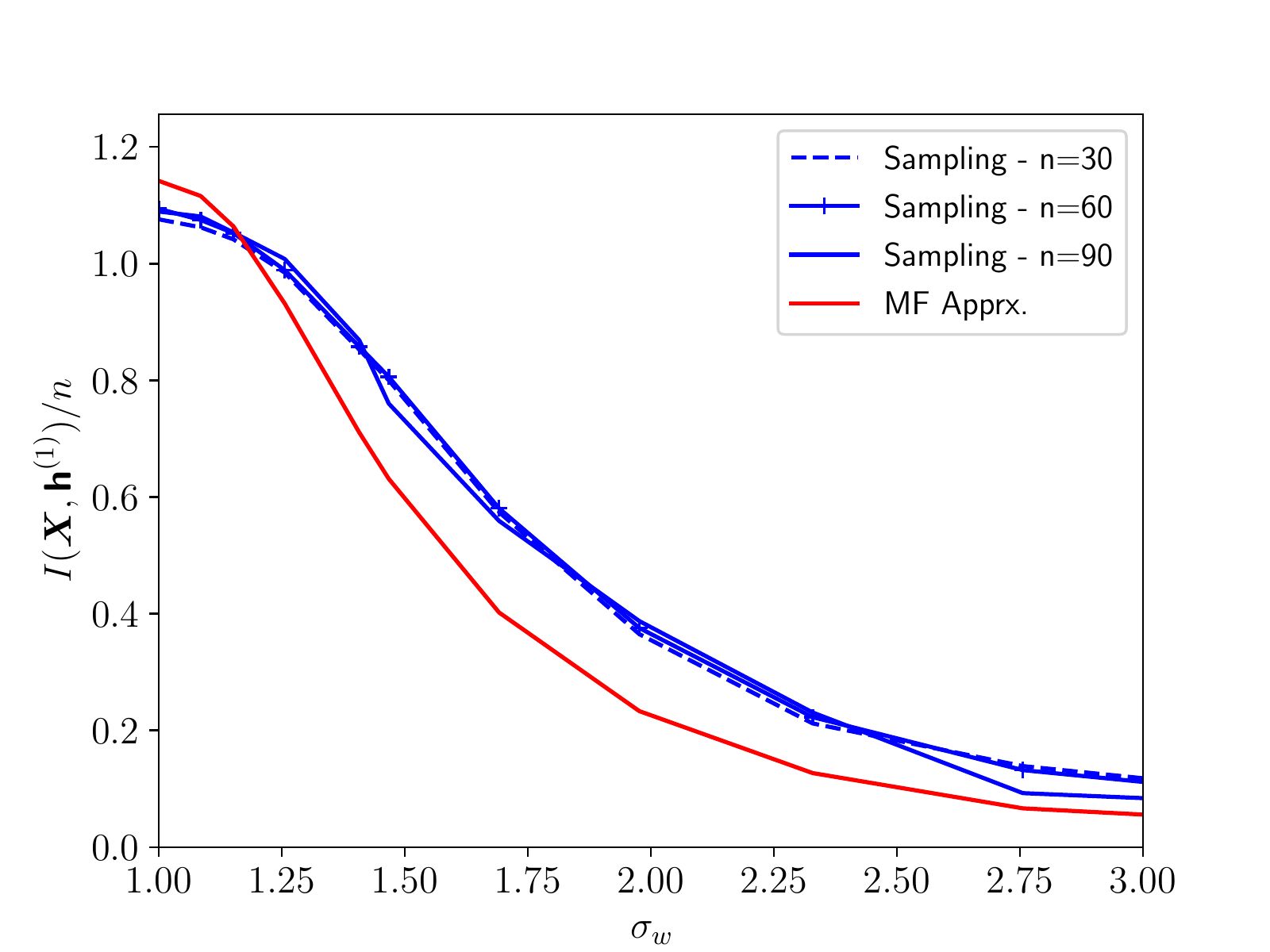}
    \caption{Mutual Information at Layer 1}
    \label{fig:f1_2}
  \end{subfigure}\\
  \centering
  \begin{subfigure}{0.4\textwidth}
    \includegraphics[width=\textwidth,trim={0.22cm 0.1cm 1.4cm 1cm},clip]{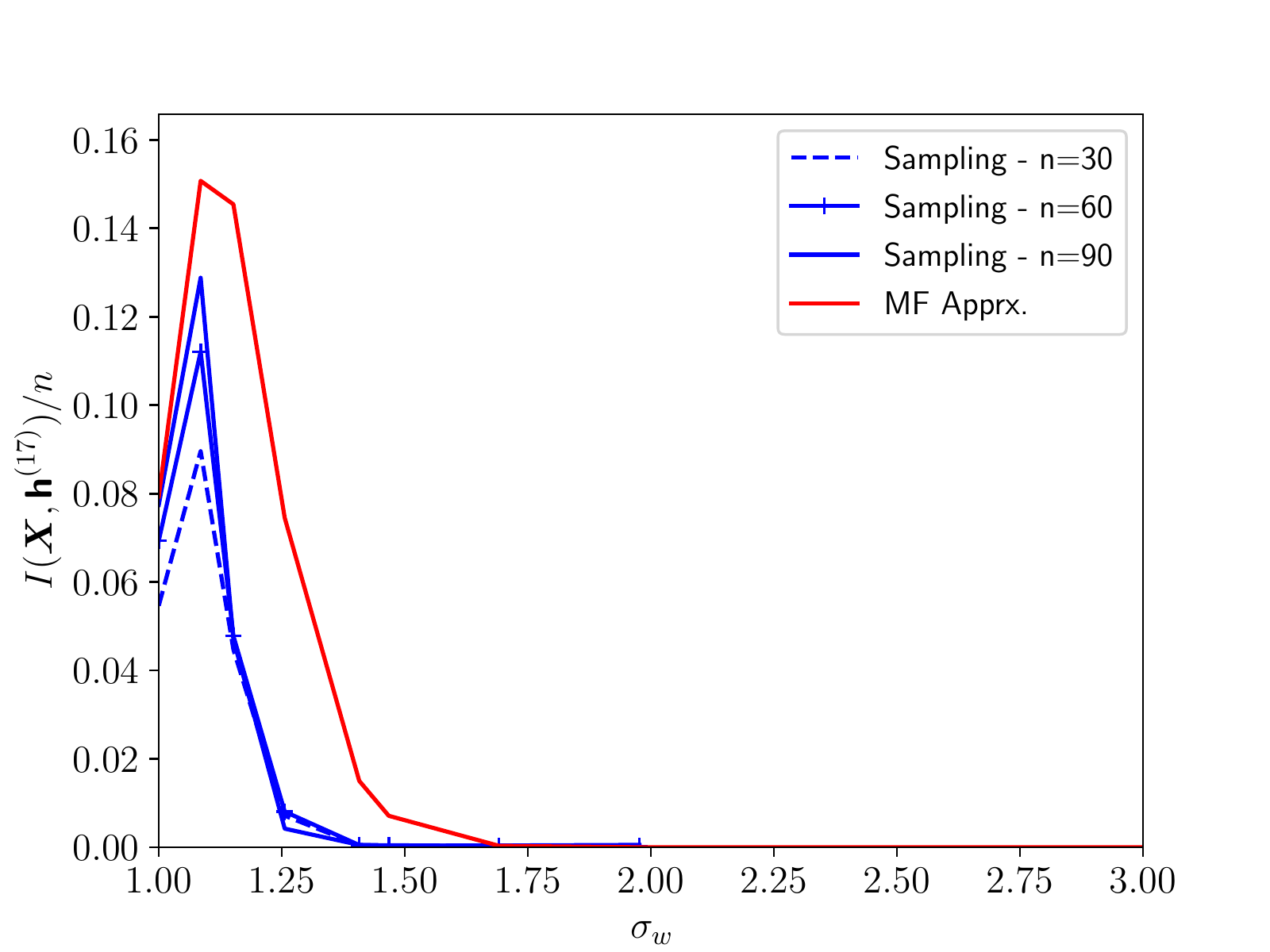}
    \caption{Mutual Information at Layer 17}
    \label{fig:f5_2}
  \end{subfigure}%
  \centering
  \caption{Comparison of the analytic lower bound \eqref{eq:lower_bound} of $I(\textbf{X},\textbf{h}^{(\ell)})/n$ and the sampled Gaussian lower bound in \eqref{eq:e1}-\eqref{eq:e2} for $n=30,60,$ and $90$ for layers $1$ (a) and $17$ (b).}
  \label{fig:increase_of_n}
\end{figure}

Figure \ref{fig:approx} compares the mean field approximation from Prop. \ref{prop:mf_lowerbound} with the direct sampling approach as described in Figure \ref{fig:num_gaussian_lb}.
We observe good general agreement, in particular at the locations where the MI is maximised.  Figure \ref{fig:increase_of_n} illustrates the sampling and mean field calculations for varying widths $n={30,60,90}$ as well as at layers $\ell=1$ and $17$. Improved agreement is observed for later layers and increased DNN width as can be expected since the mean field analysis is increasingly accurate in these limits.

\section{Comparison of MI Approximation} \label{sec:Gabrie}

The work in \cite{Abrol2020} showed that among different MI approximations as \cite{kolchinsky2018caveats} and \cite{noshad2019scalable}, the replica formula in \cite{gabrie2019entropy} is the most consistent measure with the arguments in \cite{saxe2019information}, as it models with the decrease in MI for large enough variance $\sigma_w^2$ the loss in expressivity of the DNN due to the saturation of the $tanh$ activation function. Figure \ref{fig:comp_with_Gabrie} compares the replica formula to the mean field lower bound for DNNs where the bias is null $\sigma_b=0$, the noise has a variance $\sigma_n^2 = 10^{-5}$, and the input has i.i.d. normal elements, since the replica formula is applicable only on these DNNs. 
The results show that the replica formula and the approximated lower-bound are both maximised for a value of the standard deviation $\sigma_w$ thus supporting the saturation argument in \cite{saxe2019information}. However, there is an inconsistency for low standard deviations $\sigma_w$ due to the approximation of both methods.

\begin{figure}[t!]
    \centering
    \includegraphics[width=0.4\textwidth,trim={0.5cm 0.1cm 1.6cm 1cm},clip]{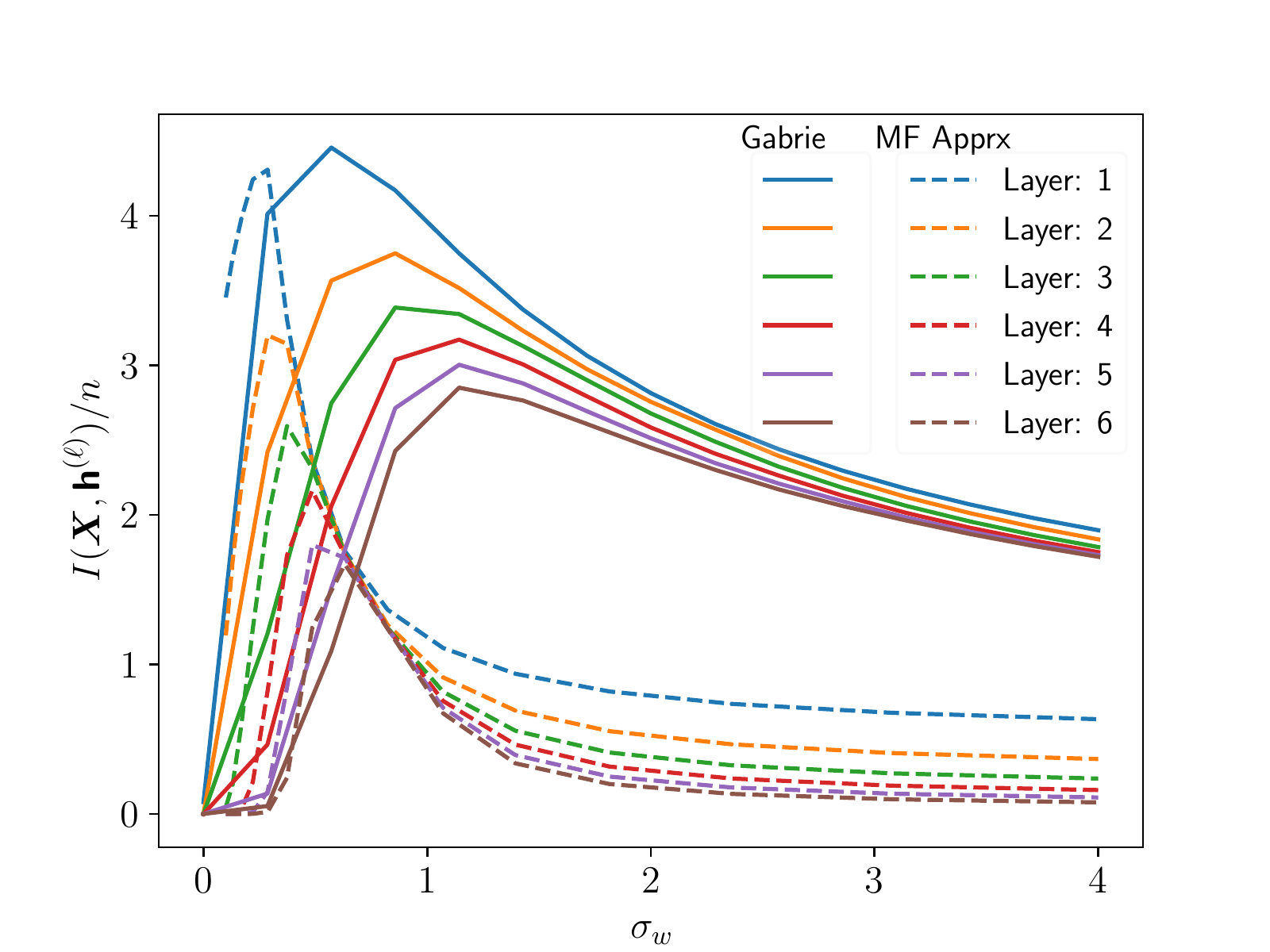}
\caption{Comparison between the \cite{gabrie2019entropy} approximation of $I(\textbf{X},\textbf{h}^{(\ell)})/n$ and the lower bound obtained with mean field approximation for a DNN with $n=1000$ when the $\sigma_w$ is changed on the x axis and $\sigma_b$ is kept fixed.}
\label{fig:comp_with_Gabrie}
\end{figure}

Finally, in \cite{Abrol2020} the analysis with the replica formula suggested that the MI information converges to a non-trivial limit as the depth increases. Figure \ref{fig:conv_mi} shows a consistent behaviour where the mean field lower bound also converges to a maximum for $\sigma_w=1$ as the depth increases, with $\sigma_b=0$ and $\phi(\cdot) = tanh(\cdot)$. Since $(\sigma_w, \sigma_b)=(1,0)$ is on the EoC for the $tanh$ activation function, these calculations suggests that the initialisations on the EoC, which considers a large depth limit, are preferable both for optimisation \cite{Pennington2018} and MI.

\begin{figure}[t!]
    \centering
    \includegraphics[width=0.4\textwidth,trim={0.5cm 0.1cm 1.6cm 1cm},clip]{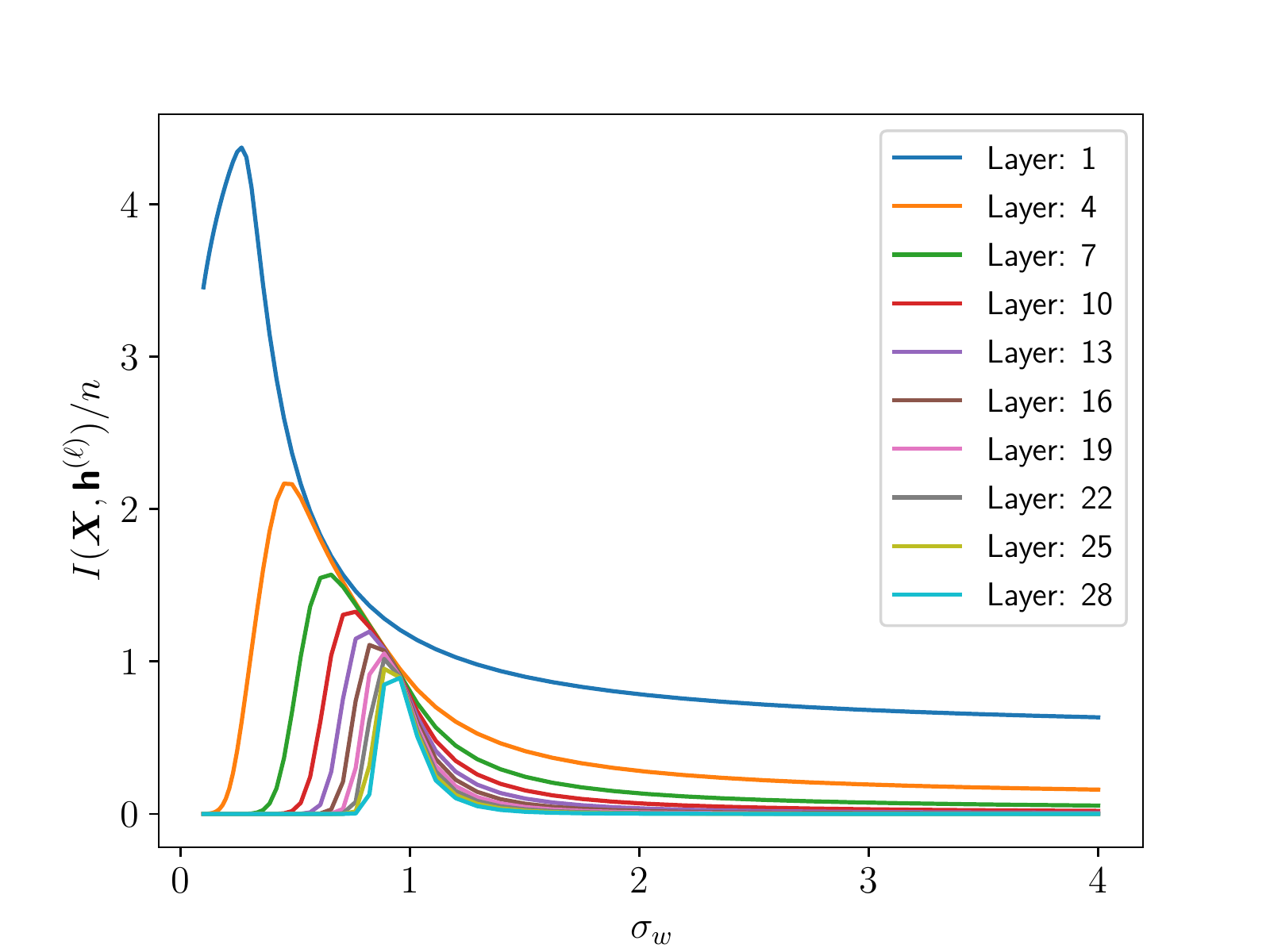}
\caption{Convergence of the lower bound of $I(\textbf{X},\textbf{h}^{(\ell)})/n$ obtained with the mean field approximation when we change the variance $\sigma_w$ and keep $\sigma_b=0$.}
\label{fig:conv_mi}
\end{figure}

\section{Conclusion}

We have presented a lower bound of the MI for feed-forward DNNs and we derived an approximation with the mean field theory which numerical experiments showed to be consistent with the original measure.  The analytic lower bound approximation allows direct investigation of how the MI of DNNs change for different initialisation parameters $(\sigma_w,\sigma_b,\phi(\cdot))$. In particular, we observe that with $\phi(\cdot)=tanh(\cdot)$ activation function, the MI is maximised for $(\sigma_w\sigma_b)$ on the EoC, which suggests the EoC initialisation are similarly optimal from a MI perspective.

\section*{Acknowledgment}

This publication is based on work supported by the EPSRC Centre for Doctoral Training in Industrially Focused Mathematical Modelling (EP/L015803/1) in collaboration with New Rock Capital Management
and by the Alan Turing Institute under the EPSRC grant EP/N510129/1.
\newpage

\enlargethispage{-1.2cm} 
\IEEEtriggeratref{8}

\bibliographystyle{IEEEtran}
\bibliography{bibliography.bib}

\end{document}